\documentclass[aps,twocolumn]{revtex4}
\usepackage{amsfonts,amsthm,amsmath,amssymb, graphics, graphicx, wrapfig,  epsfig, color, dsfont, overpic, mathrsfs,subfigure}
\usepackage{epstopdf}	

\usepackage[ruled,vlined, norelsize]{algorithm2e}

\newtheorem{Theorem}{Theorem}
\newtheorem{Definition}{Definition}

\newcommand{\vect}[1]{\boldsymbol{#1}}
\newcommand{\vectornorm}[1]{\left|\left|#1\right|\right|}

\newcommand{\de}{\delta}

\newcommand{\dd}{\text{d}}

\newcommand{\vxi}{\vect{\xi}}

\newcommand{\xv}{\vect{x}}
\newcommand{\kv}{\vect{k}}

\newcommand{\vv}{\vect{v}}

\begin{document}

\title{A Fractional Diffusion Equation for an $n$-Dimensional Correlated L\'{e}vy Walk}
\author{Jake P. Taylor-King$^{1,2}$, Rainer Klages$^{3,4}$, Sergei Fedotov$^{5}$, and
Robert A. Van Gorder$^1$}
\affiliation{$^1$Mathematical Institute, University of Oxford, Oxford, OX2 6GG, UK\\ $^2$Department of Integrated Mathematical Oncology, H. Lee Moffitt Cancer Center and Research Institute, Tampa, FL, USA\\
$^3$Max Planck Institute for the Physics of Complex Systems, N\"othnitzer Str.\ 38, D-01187 Dresden, Germany\\
$^4$School of Mathematical Sciences, Queen Mary University of London, Mile End Road, London E1 4NS, United Kingdom\\ $^5$School of Mathematics, The University of Manchester, Manchester, M13 9PL, United Kingdom}

\begin{abstract}
  L\'{e}vy walks define a fundamental concept in random walk theory
  which allows one to model diffusive spreading that is faster than
  Brownian motion. They have many applications across different
  disciplines. However, so far the derivation of a diffusion equation
  for an $n$-dimensional correlated L\'{e}vy walk remained elusive.
  Starting from a fractional Klein-Kramers equation here we use a
  moment method combined with a Cattaneo approximation to derive a
  fractional diffusion equation for superdiffusive short-range
  auto-correlated L\'{e}vy walks in the large time limit, and solve
  it.  Our derivation discloses different dynamical mechanisms leading
  to correlated L\'{e}vy walk diffusion in terms of quantities that
  can be measured experimentally.
\end{abstract}

\maketitle

\section{Introduction}
For most of the last century diffusive processes were understood in
terms of Brownian motion, which describes the random-looking
flickering of a tracer particle in a fluid. This type of stochastic
dynamics is characterized by Gaussian probability density functions
(PDFs) for both the position $\vect{x}$ and the velocity $\vect{v}$ of
a moving particle by generating a mean square displacement (MSD) of an
ensemble of particles that increases linearly for large times,
$\langle \vect{x}^2 \rangle\sim t$ \cite{vanKampen_1992,Risken_1996}.
The Brownian paradigm was challenged over the past few decades due to
more refined measurement techniques reporting anomalous diffusion,
where the MSD increases nonlinearly in time, $\langle x^2 \rangle\sim
t^{\nu}$ with $\nu\neq1$ \cite{Metzler_2000,MeKl04,Klages_2008}.
Subdiffusion with an exponent $\nu<1$ has especially been found for
motion in crowded environments \cite{MeSo15}, superdiffusion with
$\nu>1$ was observed, e.g., for chaotic transport of tracer particles
in turbulent flows \cite{SWS93} as well as for foraging biological
organisms \cite{VLRS11}. The easiest way to model spreading faster
than Brownian motion is in terms of \emph{L\'{e}vy flights}
\cite{KSZ96,Klages_2008,ZDK15}: Here the step length $\ell$ is a
random variable drawn independently and identically distributed from a
fat-tailed L\'{e}vy stable PDF characterized by power law tails,
$f(\ell)\sim \ell^{-1-\xi}$ with $0<\xi<2$. Correspondingly, L\'{e}vy
flights feature infinite propagation speeds and diverging MSDs. This
motivated the formulation of \emph{L\'{e}vy walks} (LWs)
\cite{KSZ96,ZDK15}, where a particle follows straight line
trajectories under the constraint of finite velocities by
stochastically reorientating itself (possibly with intermittent
resting phases) before repeating the process.  They became an
important concept for modeling a wide range of physical processes
\cite{ZDK15}. LWs belong to the class of (generalised) \emph{velocity
  jump processes} (VJPs): Central to their description is the running
time distribution, specified by a PDF, which describes how long a
particle moves in one direction before undergoing a stochastic
reorientation event. LWs are obtained by choosing a L\'{e}vy stable
running time PDF coupled to a corresponding step length PDF by finite
velocities, which generates superdiffusion with finite MSDs.  For VJPs
where the running time PDF is exponentially distributed, in which case
the process is memoryless, or Markov, one recovers the case of normal
diffusion with a MSD that increases linearly for large times.  These
two basic VJPs are special cases of the more general class of
Continuous Time Random Walks \cite{Metzler_2000,KlSo11}.

For a Brownian particle the PDFs of position and velocity can easily
be calculated by solving the standard diffusion equation, i.e., Fick's
Second Law, and the corresponding Fokker--Planck equation. These two
equations arise as special cases of the Klein--Kramers equation, which
is a Fokker--Planck equation both in position and velocity space
\cite{vanKampen_1992,Risken_1996}. Deriving corresponding equations
for anomalous diffusion led to fractional differential equations,
where non-integer derivatives are used based on power law
repositioning kernels with infinite second moments \cite{SKB02}.
While for subdiffusion fractional diffusion equations have been
derived based on subordination or Continuous Time Random Walk theory
\cite{Metzler_2000,KlSo11}, this problem turned out to be much more
non-trivial for L\'{e}vy walkers due to the spatio--temporal coupling
imposing finite velocities \cite{ZDK15}. Only very recently progress
was made by deriving an integro--differential wave equation for a
one-dimensional LW \cite{Fed15}. More generally, in position and
velocity space a fractional Klein--Kramers equation containing an
$n$-dimensional correlated LW as a special case was given in
\cite{Friedrich_2006, Friedrich_2006_2} without resting phases, and in
\cite{TaylorKing_2015} when resting phases are included. An open
question, however, is how to extract a fractional diffusion equation
for LWs from such a generalised Klein--Kramers equation. A key problem
for establishing a relationship between LWs and a fractional diffusion
equation is that generally a variety of fractional Laplace operators
is available, and the correct choice is not obvious. Obtaining such an
equation enables one to analytically solve first passage and arrival
time problems, which is relevant to study search problems for physical
and biological processes \cite{BLMV11}.

In this paper we derive such an equation for a superdiffusive LW by
starting from the fractional Klein--Kramers equation in
Refs.~\cite{Friedrich_2006,TaylorKing_2015}, based on an expansion in
terms of moment equations. Using a Cattaneo approximation, in the
limit of large times we obtain a diffusion equation with a Dirichlet
fractional Laplacian correction term, which can be solved exactly.
That this is not merely a mathematical exercise but that our result
has profound physical meaning is demonstrated by the fact that in
previous literature exactly such equations have been written down
\emph{ad hoc} for modelling intermittent diffusive dynamics as a
superposition of Brownian motion with L\'{e}vy flights
\cite{LAM05,LKMK08,CSK11}. This dynamics has in turn been observed in
experiments measuring the movements of biological organisms
\cite{SHBB12,dJBKW13,ZoLi99,HBDF12,PMC13,Sims10,BPPMC03}.  That our
derivation includes short-range memory makes our result very relevant
to such biophysical applications.  Our systematic derivation discloses
how this type of dynamics emerges from auto-correlated LWs by yielding
exact expressions for all quantities involved in our final equation.
These expressions are in turn based on the joint PDF for the
velocities of two successive steps and the running time PDF, which can
be measured experimentally.

We proceed as follows: We first introduce generalised VJPs by briefly
reviewing previously derived equations.  We then expand the delay
kernel characterising the fractional Klein--Kramers equation defining
a LW in Laplace space for large time, before inverting the result back
to spatio--temporal variables. Next we derive and investigate the
corresponding moment equations and then close them by use of the
Cattaneo approximation. Finally we show the equivalence of the Riesz
and Dirichlet fractional Laplacians, which enables us to write down
our central result, a fractional diffusion equation modelling an
$n$-dimensional correlated LW in the large time limit.  The section
following this is devoted to numerical simulations before we summarise
and discuss our results in the final section.

\section{Generalised velocity jump equation without rests}\label{sec_math_intro}

We give a description of a generalised VJP without rests. In the case
where there is a resting phase of stochastic length included, see
Taylor-King \emph{et. al.} \cite{TaylorKing_2015}.

At time $t=0$, a biological agent chooses a direction $\vect{\theta}$ and speed $s$ at random. The agent then travels with velocity $\vv = s\vect{\theta}$ for $\tau$ units of time, where $\tau$ is also drawn from probability distribution $f_\tau$. At time $t=\tau$, the agent then instantaneously reorientates itself with a new direction and speed. The process repeats indefinitely.

This motion is governed by two primary stochastic effects. We specify
these by PDFs, as given below.
\begin{itemize}
\item[a.)] {\em Running time:} The time spent running, denoted $\tau$, is governed by the pdf $f_{\tau}(t)$, where $\int_0^\infty f_\tau (t)\dd t=1$.
\item[b.)] {\em Reorientation:} We allow velocities from one run to
  another to be correlated. We denote the velocity during the running
  phase immediately before reorientation by $\vect{v}'$ and the
  velocity immediately post-reorientation by $\vect{v}$, where
  $\vect{v}',\vect{v}\in V$, for velocity space $V\subset
  \mathds{R}^n$ in $n$ spatial dimensions. The velocity $\vect{v}$ is
  dependent on $\vect{v}'$, and is instantly selected upon entering a
  new running phase, governed by the joint pdf
  $T(\vect{v},\vect{v}')$. We assume that this reorientation pdf is
  separable, so that $T(\vect{v},\vect{v}') = g(\vect{\theta},
  \vect{\theta}')h(s, s')/s^{n-1}$ where $\vect{\theta}$ is a vector
  of length $(n-1)$ containing angles and $s = \vectornorm{\vect{v}}$
  is the speed. In two dimensions, the turning kernel is decomposed as
  follows:
\begin{itemize}
\item[i.)]  The angle distribution: $g(\theta, \theta')$, requires the normalisation $\int_0^{2\pi}g(\theta, \theta')\dd \theta = 1$.
\item[ii.)] The speed distribution: $h(s, s')$, requires the normalisation $\int_0^{\infty}h(s, s')\dd s = 1$.
\end{itemize}
\end{itemize}

In Appendix \ref{app_gill} we describe a simple Gillespie algorithm
for generating a sample path. We remark that a VJP where velocities
from one run to another are not correlated yields a conventional LW as
a special case \cite{ZDK15}.

As derived in \cite{Friedrich_2006}, the density of particles following a velocity jump process without rests is given by 
\begin{align}
&\left[\frac{\partial}{\partial t} + \vv \cdot{\nabla_{\xv}} \right] p(t,\xv ,\vv ) \nonumber\\
& = -\int_0^t \Phi_\tau(t-s) p(s,\xv  - (t-s)\vv ,\vv ) \dd s  \nonumber \\ 
&+ \int_0^t \Phi_{\tau}(t-s)\int_V T(\vv ,\vv ')p(s,\xv  - (t-s)\vv ',\vv ')\dd\vv '\dd s \, , \label{p_forward_no_rests}
\end{align}
where $\Phi_\tau(t)$ can be found implicitly by the equation
\begin{equation}\label{eq_impractical}
\frac{\dd F_\tau}{\dd t} = - \int_0^t \Phi_\tau(s)F_\tau(t-s) \dd s \, ,
\end{equation}
for $F_\tau(t) = \int^{\infty}_t f_\tau(s)\dd s$. In Appendix \ref{app_deriv} we
offer an alternative derivation, which may appeal especially to
mathematical biologists. More practically than equation
\eqref{eq_impractical}, one can use the Laplace space description
\begin{equation}\label{eq_Phi_rel}
\bar{\Phi}_\tau (\lambda) = \frac{\lambda \bar{f}_\tau  (\lambda)}{1 - \bar{f}_\tau (\lambda)} = \frac{1 - \lambda \bar{F}_\tau (\lambda)}{\bar{F}_\tau(\lambda)} \, .
\end{equation}
An interesting point to note is that for the Markov velocity jump process where $f_\tau$ is exponentially distributed with mean $\mu = \chi^{-1}$, then $\Phi_\tau$ manifests as a constant rate parameter, so
\begin{align}\label{eq_f_exp_const_Phi_1}
f_\tau (t) = \chi e^{-\chi t} \, \iff \, \bar{f}_\tau (\lambda) = \frac{\chi}{\lambda + \chi} \, , 
\end{align}
and therefore
\begin{align}\label{eq_f_exp_const_Phi_2}
   \bar{\Phi}_\tau (\lambda) = \chi \,\iff \, \Phi_\tau (t) = \chi \delta (t) \, .
\end{align}

\section{Delay kernel behaviour}\label{sec_delay_kernel}
For us to consider the large time behaviour of the velocity jump process, we wish to explore equation \eqref{p_forward_no_rests} for large $t$. Conveniently in Laplace space, large time $t$ corresponds to small Laplace variable $\lambda$. We now study the form of delay kernel $\Phi_\tau$ given a running distribution $f_\tau$. By definition of $\Phi_\tau$ in Laplace space, we need to investigate equation \eqref{eq_Phi_rel}.

In the case where the mean and variance of the running distribution is finite, one can simply Taylor expand the underlying distribution $f_\tau$ and delay kernel $\Phi_\tau$ in Laplace space. When either of the first two moments are undefined, one must be more careful and rely on asymptotic expansions. The case where both the mean and variance are finite is explored in \cite{TaylorKing_2015}. In this paper we review this case, and present analysis for when the variance is infinite. The case when both the mean and variance of the running distribution is infinite is more difficult and will be investigated in future work. We expand the running distribution $f_\tau$ in Laplace space including terms up to quadratic order.

\subsection{Finite mean and variance}\label{sec_fin_mean_fin_var}
If the first two moments of $f_\tau$ are defined then
\begin{eqnarray}
\bar{f}_\tau (\lambda)  &=& 1 - \langle \tau \rangle \lambda + \frac{1}{2}  \langle \tau^2 \rangle \lambda^2 - ...   \\
&=& 1 - \mu \lambda + \frac{1}{2} \left( \sigma^2 + \mu^2 \right)\lambda^2 - ... \, ,
\end{eqnarray}
where $\mu = \mathds{E}(\tau) = \langle \tau \rangle$, and $\sigma^2 = \mathds{E}(\tau^2) - [\mathds{E}(\tau) ]^2 = \langle \tau^2 \rangle - \langle \tau \rangle^2$. Taylor expanding $\Phi$, we find that
\begin{eqnarray}\label{eq_Phi_Taylor_ex}
\bar{\Phi}(\lambda) &=& \bar{\Phi}(0) + \bar{\Phi}'(0)\lambda + \frac{1}{2}\bar{\Phi}''(0)\lambda^2 + ... \text{ as }\lambda\to0 \, ,\\
&=& \frac{1}{\mu} +  \frac{1}{2}\left(\frac{\sigma^2}{\mu^2}-1 \right)\lambda + \mathcal{O} \left(  \lambda^2 \right) \text{ as }\lambda\to0 \, . \label{eq_Phi_finite_mean_var}
\end{eqnarray}
Note that equation \eqref{eq_Phi_finite_mean_var} is consistent with equations \eqref{eq_f_exp_const_Phi_1}--\eqref{eq_f_exp_const_Phi_2} as when $f_\tau$ is exponentially distributed $\mu^2 = \sigma^2$. This case was examined in detail in earlier work \cite{TaylorKing_2015}, and eventually leads to a diffusion equation in the large time limit.

\subsection{Finite mean, infinite variance}\label{sec_fin_mean_inf_var_1}
If only the first moment is defined (i.e. finite mean, infinite variance), then we observe a fat tailed distribution of the form
\begin{equation}
f_\tau (t) \sim t^{-2- \alpha} \text{ as }t\to\infty \, ,
\end{equation}
for $\alpha\in(0,1]$. In Laplace space, this gives the expansion \cite{Portillo_2011}
\begin{equation}\label{eq_f_tau_Laplace_exp_fin_mean}
\bar{f}_\tau (\lambda)= 1 - \langle \tau \rangle \lambda + \gamma \lambda^{1 + \alpha} - ...  =  1 - \mu \lambda + \gamma  \lambda^{1 + \alpha} - ... 
\end{equation}
as $\lambda\to0$, where $\mu$ and $\gamma$ will depend on the parameters of the distribution $f_\tau (t)$. In Appendix \ref{app_Pareto}, we give an example expansion of $\bar{f}_\tau$ for the Pareto distribution.

Using equation \eqref{eq_Phi_rel} in conjunction with the expansion given by equation \eqref{eq_f_tau_Laplace_exp_fin_mean}, we find that
\begin{equation}
\bar{\Phi}(\lambda) \sim \frac{\lambda (1 - \mu \lambda +\gamma  \lambda^{1 + \alpha} )}{  \mu \lambda - \gamma  \lambda^{1 + \alpha}} \sim \frac{1 - \mu \lambda +\gamma  \lambda^{1 + \alpha} }{  \mu - \gamma  \lambda^{ \alpha}} \, .
\end{equation}
Noting the Geometric expansion 
\begin{equation}
\frac{1}{ \mu - \gamma  \lambda^{ \alpha} } = \frac{1}{\mu} \sum_{n=0}^{\infty} \left( \frac{ \gamma \lambda^{ \alpha}   }{\mu} \right)^n = \frac{1}{\mu} +  \frac{ \gamma \lambda^{ \alpha}   }{\mu^2} +  \frac{( \gamma \lambda^{ \alpha} )^2  }{\mu^3} ...\, ,
\end{equation}
for $| \lambda^{ \alpha} | < \mu / \gamma$, therefore
\begin{equation}\label{eq_Phi_finite_mean_inf_var}
\bar{\Phi}(\lambda) \sim \frac{1}{\mu} +  \frac{ \gamma \lambda^{ \alpha}   }{\mu^2} - \lambda + \mathcal{O} \left(  \lambda ^{ \min\{ {1  + \alpha },\, {2\alpha } \} } \right)  \text{ as }\lambda\to0 \, .
\end{equation}
If $\alpha =1$ and $\gamma = (\sigma^2 + \mu^2)/2$, then equation \eqref{eq_Phi_finite_mean_inf_var} is consistent with equation \eqref{eq_Phi_Taylor_ex}.

\section{Inversion back to spatio--temporal variables}\label{sec_inversion}

We wish to carry out the analysis for the velocity jump process without rests. The density is given by equation \eqref{p_forward_no_rests}. In Fourier--Laplace space, we have
\begin{align}\label{eq_p_forward_Fourier_Laplace}
& \left[\lambda + i\kv\cdot\vv  \right] \tilde{\bar{p}}(\lambda, \kv ,\vv ) - \tilde{p}_0(\kv ,\vv ) = -\bar{\Phi}_\tau (\lambda + i\kv\cdot\vv )\tilde{\bar{p}}(\lambda,\kv ,\vv )\nonumber \\ &+\bar{\Phi}_\omega(\lambda + i\kv\cdot\vv  ) \int_V T(\vv ,\vv ')\tilde{\bar{p}}(\lambda, \kv ,\vv ') \dd \vv '.
\end{align}
It is now required that we analyse the term
\begin{equation}\label{eq_form_to_eval}
\bar{\Phi}(\lambda + i\kv\cdot\vv ) \tilde{\bar{p}}(\lambda, \kv, \vv ) \, .
\end{equation}
When we have finite mean and infinite variance, we wish to evaluate the expansion given
by equation \eqref{eq_Phi_finite_mean_inf_var}. Therefore we wish to
evaluate terms of the form
\begin{equation}\label{eq_expand_fin_mean_inf_var}
\bar{\Phi}(\lambda + i\kv\cdot\vv ) \sim \left[ \frac{1}{\mu}+\frac{\gamma  (\lambda + i\kv\cdot\vv )^{\alpha} }{\mu^2}  - (\lambda + i\kv\cdot\vv )  \right]  \, .
\end{equation}
The term $(\lambda + i\kv\cdot\vv )^{\alpha}$ is a multidimensional
version of the fractional material derivative introduced in
\cite{Sokolov_2003}. Because the fundamental solution of the material
derivative equation is only defined in a weak sense
\cite{Jurlewicz_2012}, and limited analytic progress made in
dimensions higher than one \cite{Magdziarz_2015}, we avoid using this
pseudo-differential operator. If we are considering the large time
limit, then we make the ansatz $\lambda = \mathcal{O}( || \kv ||^{1 +
  \alpha} )$. In this regime we can then use the Binomial Theorem to
simplify equation \eqref{eq_expand_fin_mean_inf_var} if we specify
that $|| \vv || \sim 1$. We obtain the term
\begin{equation}\label{eq_BinThm_expansion}
\left[ \frac{1}{\mu}+\frac{\gamma  ( [i\kv\cdot\vv ]^{\alpha} + \alpha \lambda [i\kv\cdot\vv ]^{\alpha-1} + \dots ) }{\mu^2}  - (\lambda + i\kv\cdot\vv )  \right]  \, .
\end{equation}
We henceforth drop the term $\lambda[i\kv\cdot\vv ]^{\alpha-1}$; this is because this term will only be large in the small region when $|| \kv\cdot\vv || \ll 1$.  

By inserting the expression given in equation \eqref{eq_BinThm_expansion} into \eqref{eq_expand_fin_mean_inf_var}, putting it into \eqref{eq_p_forward_Fourier_Laplace} and then inverting into spatio--temporal variables, we can write down equation \eqref{p_forward_no_rests} with the delay kernel that relates to a running distribution with finite mean and infinite variance
\begin{equation}
\begin{aligned}
& \left[\frac{\partial}{\partial t} + \vv \cdot{\nabla_{\xv }} \right]  p(t,\xv ,\vv ) =  \\
& - \left[ \frac{1}{\mu}+\frac{\gamma }{\mu^2}  (\vv \cdot\nabla_{\xv })^{\alpha}  -  \left( \frac{\partial}{\partial t}  + \vv \cdot\nabla_{\xv } \right)  \right] p(t, \xv , \vv )   \\ 
&+ \int_V T(\vv ,\vv ')  \left[ \frac{1}{\mu}+\frac{\gamma }{\mu^2} (\vv '\cdot\nabla_{\xv })^{\alpha} -  \left( \frac{\partial}{\partial t}  + \vv '\cdot\nabla_{\xv } \right)  \right] \\
& \qquad \qquad\qquad \qquad\qquad \qquad \times p (t, \xv , \vv ') \dd\vv ' \, .
\end{aligned}
\end{equation}

\section{Moment equations}\label{sec_moment_eqs}
When integrating over the velocity space, we generate an equation for the conservation of mass; this equation refers to the flux of the momentum. More generally, by considering successively greater monomial moments in the velocity space, one obtains a system of $k$ equations where the equation for the time evolution of moment $k$ corresponds to the flux of moment $k+1$. It therefore becomes necessary to `close' the system of equations to create something mathematically tractable. We define the notation for the first three moments as
\begin{align}
m^0 = &\int_V  p (t, \xv , \vv ) \dd \vv \, ,\, \vect{m}^1= \int_V \vv  p (t, \xv , \vv ) \dd \vv \, , \nonumber \\
& \text{ and } M^2 = \int_V \vv \vv^T  p (t, \xv , \vv ) \dd \vv \, .
\end{align}
In order to make progress, we must first make an assumption on the turning kernel $T$. By considering that the mean post-turn velocity has the same orientation as the previous velocity, we define the index of persistence $\psi_d$ via the relation
\begin{equation}
\bar{\vect{v}}(\vect{v}') = \int_V \vect{v}T(\vect{v},\vect{v}') \dd \vect{v} = \psi_d \vect{v}'.
\end{equation}
Informally, this means that turning angles between consecutive velocities have zero mean.

Integrating over $\vv$, our equation for $m^0$ is just conservation of mass
\begin{equation}
\frac{\partial m^0}{\partial t} + {\nabla_{\xv }}\cdot \vect{m}^1 = 0  \, .
\end{equation}
For $\vect{m}^1$, our equation becomes
\begin{equation}\label{fin_mean_inf_var_before_Cat}\begin{aligned}
& \psi_d \left[\frac{\partial \vect{m}^1}{\partial t} + \nabla_{\xv}\cdot M^2\right] \\
& = {-(1 - \psi_d )} \left[\frac{\vect{m}^1}{\mu} + \frac{\gamma  }{\mu^2}  \int_V \vv (\vv\cdot\nabla_{\xv})^\alpha \, p(t,\xv,\vv)\, \dd \vv   \right]  \, .
\end{aligned}\end{equation}
One option is to use the Cattaneo approximation to make progress.

\section{Cattaneo approximation step}\label{sec_cattaneo}

We need to explore our options for methods to close the velocity space. In the velocity jump literature, arguably the most cogent method is the Cattaneo approximation popularised by Hillen \cite{Hillen_2003, Hillen_2004, TaylorKing_2015}. 

For the case where the speed distribution is independent of the previous running step, i.e., $h(s,s') = h(s)$, we approximate $M^2$ by the second moment of some function $u_\text{min} = u_\text{min} (t,\xv ,\vv )$, such that $u_\text{min}$ has the same first two moments as $p = p(t, \xv , \vv )$ and is minimised in the $L^2(V)$ norm weighted by $h(s)/s^{n-1}$. This is essentially minimising oscillations in the velocity space whilst simultaneously weighting down speeds that would be unlikely to occur \cite{Hillen_2003}. The Cattaneo approximation is particularly valid for large times, as it assumes that any initial data in the velocity space has been smoothed out.

We introduce Lagrangian multipliers $\Lambda^0 = \Lambda^0(t, \xv )$ and $\vect{\Lambda}^1 =\vect{\Lambda}^1(t, \xv )$ and then define
\begin{equation}\begin{aligned}
H(u) := & \frac{1}{2} \int_V \frac{u^2}{h(s)/s^{n-1}}\dd \vv  - \Lambda^0\left(\int_V u \dd \vv  - m^0 \right) \\ & \qquad  - \vect{\Lambda}^1\cdot \left(\int_V\vv  u \dd \vv  - \vect{m}^1 \right) \, .
\end{aligned}\end{equation}
By the Euler--Lagrange equation \cite{Gregory}, we can minimise $H(u)$ to find that
\begin{equation}
u(t, \xv , \vv )  = \frac{\Lambda^0(t, \xv ) h(s)}{s^{n-1}} + \frac{(\vect{\Lambda}^1(t, \xv ) \cdot \vv ) h(s)}{s^{n-1}} \, .
\end{equation}
We now use the constraints to find $\Lambda^0$ and $\vect{\Lambda}^1$. For $m^0$ we have
\begin{equation}
m^0 = \int_V u\, \dd \vv  =  \Lambda^0 \int_V  \frac{h(s)}{s^{n-1}} \dd \vv   = \Lambda^0 A_{n-1} \, ,
\end{equation}
where $A_n = \text{Area}( \mathds{S}^{n} )$ and $\mathds{S}^{n} = \{\xv \in\mathds{R}^{n+1}:\vectornorm{\xv } = 1\}$ is the hollow $n$-sphere centred at the origin. Notice also that the $\int_V \vv h(s)/s^{n-1} \dd \vv  = \vect{0}$ by symmetry. For the first moment, we calculate
\begin{equation}
\vect{m}^1 = \int_V \vv\,  u\, \dd \vv  =   \vect{\Lambda}^1 \cdot \int_V   \frac{h(s) \vv \vv ^T}{s^{n-1}} \dd \vv   = S_h^2 V_n \vect{\Lambda}^1  \, ,
\end{equation}
where $S_h^{\beta} = \int_0^\infty s^{\beta} h(s) \dd s$ is the $\beta^{\text{th}}$ moment of the speed distribution $h$, and $V_n = \text{Vol}(\mathds{V}^{n})$ where $\mathds{V}^n$ is the closure of $\mathds{S}^{n-1}$, i.e., the filled unit ball around the origin. Therefore, we can stipulate the form for $u_{\text{min}}$ as
\begin{equation}\label{eq_u_min_form}
u_{\text{min}}(t, \xv , \vv )  = \frac{m^0 (t, \xv ) h(s)}{s^{n-1}A_{n-1}} + \frac{(\vect{m}^1 (t, \xv ) \cdot \vv ) h(s)}{S_h^2 s^{n-1}V_n}   \, .
\end{equation}

We now approximate the second moment of $p$ by the second moment of $u_{\text{min}}$. Noting that $ {A_{n-1}} / {V_n}  = n$,
\begin{equation}
M^2(u_{\text{min}}) = \int_V \vv \vv ^T u_{\text{min}}(t,\xv ,\vv )\dd \vv  =  \frac{S_h^2}{n}I_n m^0 (t, \xv ) \,  .
\end{equation}
In the non-fractional case, this allows us to close the set of moment equations by the approximation $\nabla_{\xv } \cdot M^2 \approx \frac{S_h^2}{n}\nabla_{\xv }m^0$. 

In the fractional case, one obtains terms of the form $\int_V \vv (\vv
\cdot \nabla_{\xv})^\beta p\, \dd \vv$, to which one must evaluate.
We do so in Appendix \ref{app_riesz} by defining the Riesz
  derivative, which is needed to perform these evaluations.

\section{Effective fractional diffusion equation}\label{sec_eff_frac_diff_eq}

For $\vect{m}^1$, using the Cattaneo approximation, equation \eqref{fin_mean_inf_var_before_Cat} becomes
\begin{equation}\begin{aligned}
& -\frac{\psi_d}{1 - \psi_d} \left[\frac{\partial \vect{m}^1}{\partial t} + \frac{S_h^2}{n} {\nabla_{\xv }} m^0\right] \\
& \qquad = \frac{\vect{m}^1}{\mu} + \frac{\gamma  }{\mu^2}   \left( \frac{ S_h^{1 + \alpha }}{ A_{n-1}  }  \nabla_{M}^{\alpha} m^0 + \frac{ S_h^{2 + \alpha  } }{V_n S_h^2 } J^{\alpha}_M \vect{m}^1 \right)     \, .
\end{aligned}\end{equation}
We can eliminate for $m^0$ to obtain
\begin{equation}\begin{aligned}
&\frac{\mu \psi_d}{(1 - \psi_d )}  \underbrace{ \frac{\partial^2 m^0}{\partial t^2} }_{\mathcal{O}(\lambda^2)} + \underbrace{ \frac{\partial m^0}{\partial t} }_{\mathcal{O}(\lambda)}  =   \frac{\mu \psi_d}{(1 - \psi_d )}\frac{S_h^2}{n} \underbrace{ {\nabla_{\xv }^2} m^0 }_{\mathcal{O}(|| \vect{k} ||^2)} \\
& \qquad\qquad+   \frac{\gamma  S_h^{\alpha + 1 }  }{\mu   A_{n-1} }  \underbrace{ \mathds{D}^{1 + \alpha}_M  m^0 }_{\mathcal{O}(|| \vect{k} ||^{1 + \alpha})} +  \frac{\gamma   S_h^{2 + \alpha  }  }{\mu V_n S_h^2}  \underbrace{ \mathds{D}^{1 + \alpha}_M \vect{m}^1 }_{\mathcal{O}(\lambda || \vect{k} ||^{1 + \alpha})}  \, .
\end{aligned}\end{equation}
We are using the scaling $\lambda = \mathcal{O}(|| \kv ||^{1 + \alpha} )$ and we can ignore higher order terms. Therefore, our equation becomes
\begin{equation}\label{eq_penultimate_finite_mean}
 \frac{\partial m^0}{\partial t}  =   \frac{\mu \psi_d}{(1 - \psi_d )}\frac{S_h^2}{n}  {\nabla_{\xv }^2} m^0  +   \frac{\gamma S_h^{\alpha + 1 }  }{\mu   A_{n-1} }   \mathds{D}^{1 + \alpha}_M  m^0 \, .
\end{equation}
In order to foster tractability, we can relate fractional derivatives via the following theorem.

\begin{figure}
\begin{center}
\includegraphics[scale=.3]{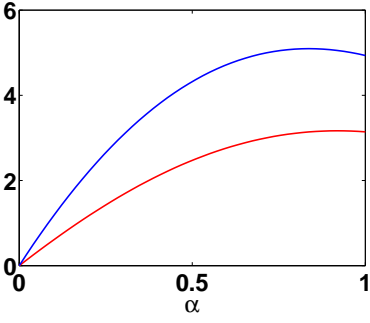}
  \end{center}
  \caption{{Plot of the function $\Upsilon_d(\alpha)$ for $0<\alpha<1$ with $d=2$ (red) and $d=3$ (blue). \label{fig_I_d_alph_plot}}}
\end{figure}

\begin{Theorem}\label{Thm_relate_frac_deriv}
The Riesz fractional derivative given by equation \eqref{eq_Riesz_laplacian} can be related to the Dirichlet fractional Laplacian for mixing measure $M(\dd \vect{\theta} ) = \dd \vect{\theta}$ via the relationship
\begin{equation}
\mathds{D}^{1 + \alpha}_M f(\xv) \equiv \Upsilon_d(\alpha) \Delta^{\frac{1 + \alpha}{2}}_{\xv} f(\xv) \, ,
\end{equation}
where the Dirichlet fractional Laplacian is defined via the Fourier transform as
\begin{equation}\label{eq_def_frac_Laplace}
\mathscr{F} \left( \Delta^{\beta}_{\xv} f(\xv) \right) = - || \kv ||^{2\beta} f(\kv) \, ,
\end{equation}
and
\begin{equation}\label{eq_I_d_alpha}
 \Upsilon_d(\alpha) = 
\left\{ \begin{array}{cc}
2\sqrt{\pi} \sin (\frac{\alpha \pi}{2}) \frac{ \Gamma (\frac{2 + \alpha}{2})}{\Gamma (\frac{3 + \alpha}{2})} & \quad d = 2 \\
4 \pi^{3/2} \sin (\frac{\alpha \pi}{2}) \left(  \frac{ \Gamma (\frac{2 + \alpha}{2})}{\Gamma (\frac{3 + \alpha}{2})} -  \frac{ \Gamma (\frac{4 + \alpha}{2})}{\Gamma (\frac{5 + \alpha}{2})} \right) & \quad d = 3 
 \end{array} \right. \, .
\end{equation}
Crucially, $\Upsilon_d(\alpha)$ depends on the dimension $d$, is a non-monotonic function of $\alpha$, and therefore has a maximum or minimum $\alpha^*_d$ where $\frac{\dd }{\dd \alpha} \Upsilon_d(\alpha^*_d) = 0$. An illustration of $\Upsilon_d(\alpha)$ is shown in Fig. 1 for $0<\alpha<1$ for $d=2,3$.
\end{Theorem}
\begin{proof}
See Appendix \ref{app_proof}.
\end{proof}
We remark that one could conceivably consider another fractional
  Laplacian, depending upon the derivation and application to be
  studied. Our choice is motivated by computational tractability.

Using Theorem \ref{Thm_relate_frac_deriv}, one can rewrite equation \eqref{eq_penultimate_finite_mean} as
\begin{equation}\label{eq_final_finite_mean}
 \frac{\partial m^0}{\partial t}  =   \frac{\mu \psi_d}{(1 - \psi_d )}\frac{S_h^2}{n}  {\nabla_{\xv }^2} m^0  +   \frac{\gamma S_h^{\alpha + 1 }  }{\mu   A_{n-1} }    \Upsilon_d(\alpha) \Delta^{\frac{1 + \alpha}{2}}_{\xv}  m^0 \, .
\end{equation}
Equation \eqref{eq_final_finite_mean} is our central result. To solve this equation, we make use of the following theorem from Blumenthal and Getoor \cite{Blumenthal_1960}.

\begin{Theorem}
For the transition density $f_\beta (t,\xv)$ defined by
\begin{equation}
e^{-t || \vxi ||^{\beta}} = \int_{\mathds{R}^n} e^{-i(\xv \cdot \vxi)} f_\beta (t,\xv) \dd \xv \, ,
\end{equation}
then for $0<\beta\leq 2$, $f$ is given by the self-similarity relation $f_\beta (t,\xv) = f_\beta (1,\xv / t^{1/\beta}) / t^{n/\beta}$ and
\begin{equation}
f_\beta (1,\xv) = \frac{1}{(2\pi)^{n/2} || \xv ||^{n/2 - 1}  } \int_0^{\infty} e^{-s^{\beta}} s^{n/2} J_{\frac{n-2}{2}}(||\xv || s )\dd s\, ,
\end{equation}
where $J_n(x)$ is the $n^{\text{th}}$ Bessel Function of the First Kind. When $\beta=2$, we obtain the usual Gaussian distribution.
\end{Theorem}
\begin{proof}
See Blumenthal and Getoor \cite{Blumenthal_1960}.
\end{proof}

Therefore the solution to equation \eqref{eq_final_finite_mean} is
\begin{equation}\label{eq_final_finite_mean_sol}
m^0(t,\xv) = f_2 \left(\frac{t}{D_2},\xv \right) * f_{1+\alpha} \left(\frac{t}{D_{1+\alpha}},\xv \right) \, ,
\end{equation}
where $*$ represents the convolution operator and
\begin{equation}
D_2 =  \frac{\mu \psi_d}{(1 - \psi_d )}\frac{S_h^2}{n} \quad\text{and}\quad D_{1 + \alpha} = \frac{\gamma S^{\alpha + 1 }   }{\mu   A_{n-1} } \Upsilon_d(\alpha)  \, . \label{eq:diffcoeffs}
\end{equation}

We note that if $D_{1 + \alpha}$ has an interior maximum (i.e. there exists some $\alpha^*\in(0,1)$ such that
$\frac{\dd}{\dd \alpha}\left[ D_{1 + \alpha}  \right]_{\alpha = \alpha^*} = 0$ and $\frac{\dd^2}{\dd \alpha^2}\left[ D_{1 + \alpha}  \right]_{\alpha = \alpha^*} < 0$), then by examining equation \eqref{eq_final_finite_mean} in Fourier space, one can optimally reduce the order 1 modes. While $\Upsilon_d(\alpha)$ given by equation \eqref{eq_I_d_alpha} will have an interior maximum, one would need to choose a particular form of $f_\tau$ in such a way that $D_{1 + \alpha}$ has an interior maximum.

\section{Numerical simulations}\label{sec_numerics}

To demonstrate the validity of equation \eqref{eq_final_finite_mean}
and its solution equation \eqref{eq_final_finite_mean_sol}, we
reconstruct $m^0(t,\xv)$ from simulations of Algorithm
\ref{algo_vj_no_rests} defined in Appendix \ref{app_gill}. In
Figs.~2--3, we plot the distribution of sample paths released from the
origin in two dimensions for varying $D_2$ and $D_{1 + \alpha}$. As
the primary contribution of this paper is analysis, and there are many
parameters involved with our model, we avoid carrying out a detailed
numerical study. To obtain a sufficient match between simulation and
the analytic expression for $m^0$, $10^4$ simulations of Algorithm
\ref{algo_vj_no_rests} were carried out. At time $T_{\text{end}} =
10^5$, the simulations are stopped. The resulting sample paths are
binned into $300\times 300$ boxes to recreate the distribution shown
in Fig.~2; a cross section is also shown to show the close match
between distributions in Fig. 3.

The underlying choices made in running distribution $f_\tau$ and turning kernel $T$ are as follows. The running distribution $f_\tau$ is given by equation \eqref{eq_f_tau_pareto} for $\tau_0 = 1/2$ and $\beta = 3/2$ (so $\alpha = 1/2$), in which case $\mu = \langle \tau \rangle = 1$. In two dimensions we write the turning kernel $T$ as $T(\vect{v},\vect{v}') = g(\theta, \theta')h(s)/s$ and we choose
\begin{equation}
g( \theta, \theta') = \frac{e^{\kappa \cos(\theta - \theta')}}{2\pi I_0(\kappa)},\text{ and } h(s) = \de(s - s^*), \label{eq:lwconstraint}
\end{equation}
and we specify that $s^* = 1$. The angle change distribution $g$ is a von Mises distribution; the index of persistence $\psi_d$ is then $\psi_d =  I_1(\kappa) /  I_0(\kappa)$. To vary the contribution from Gaussian and fractional part of the analytic solution, we run the simulation twice, once with $\kappa = 0$ (so $D_2 = 0$ and $D_{1 + \alpha} \approx 0.493$) so to highlight the non-Gaussian nature of $m^0$, and once with $\kappa = 10$ (so $D_2 \approx 9.228$ and $D_{1 + \alpha} \approx 0.493$) for a more Gaussian-like solution. We see that we clearly obtain a close match between the simulation and analytic solution curves. As the number of simulations increases, the curves become indistinguishable by eye.

\begin{figure}
\begin{center}
\epsfxsize=4cm
\subfigure{\epsfbox{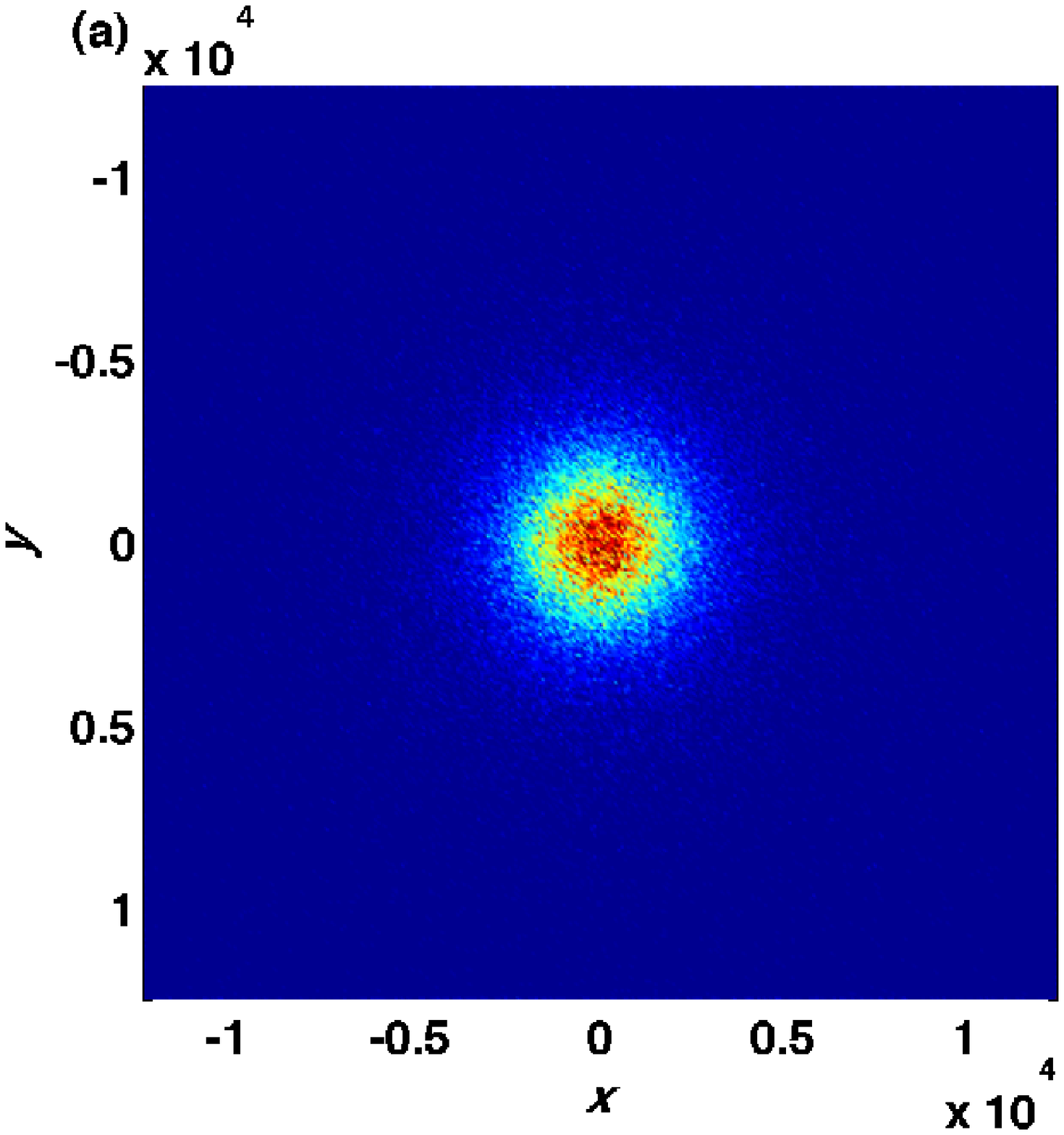}}
\epsfxsize=4cm
\subfigure{\epsfbox{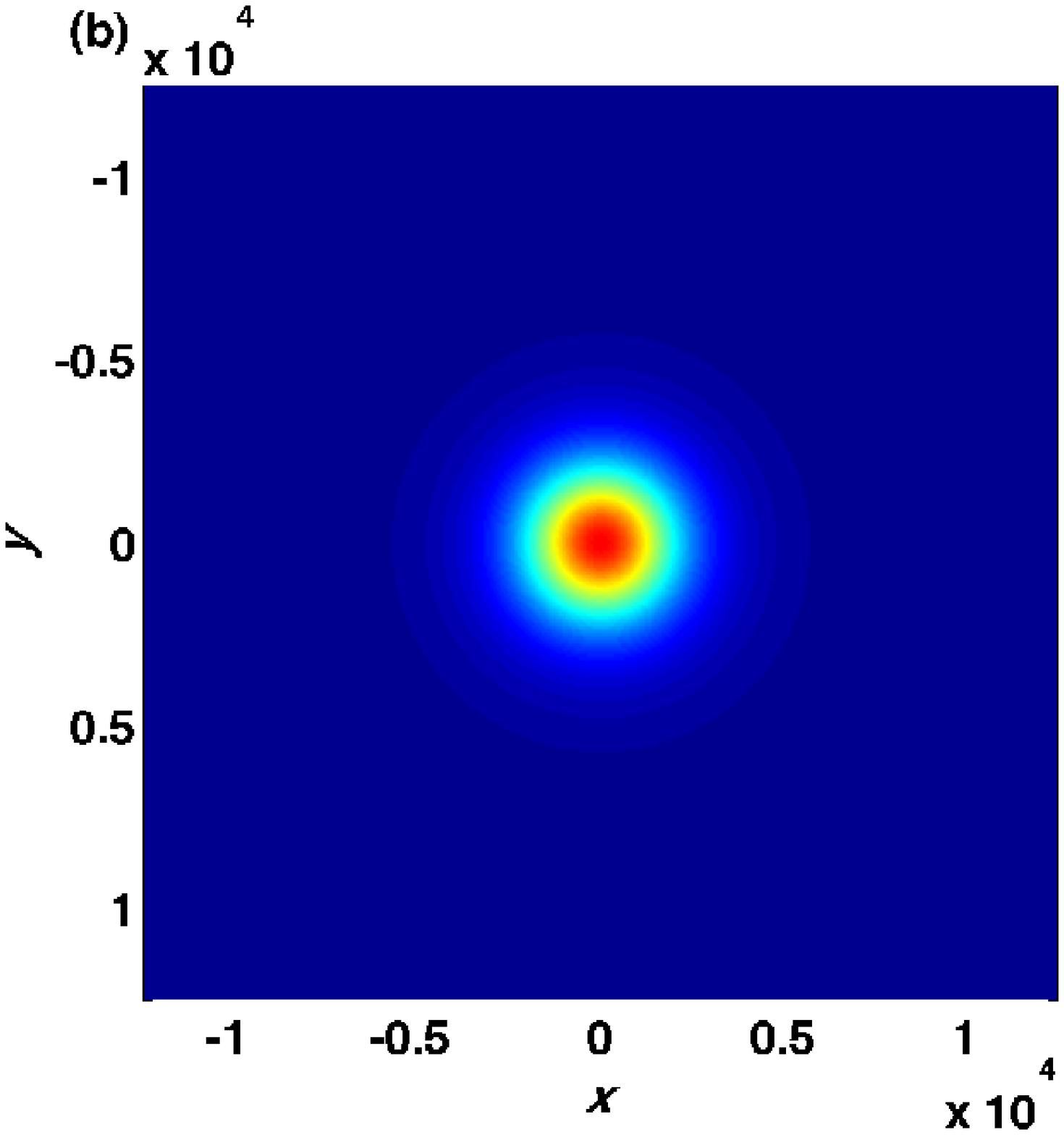}}\\
\epsfxsize=4cm
\subfigure{\epsfbox{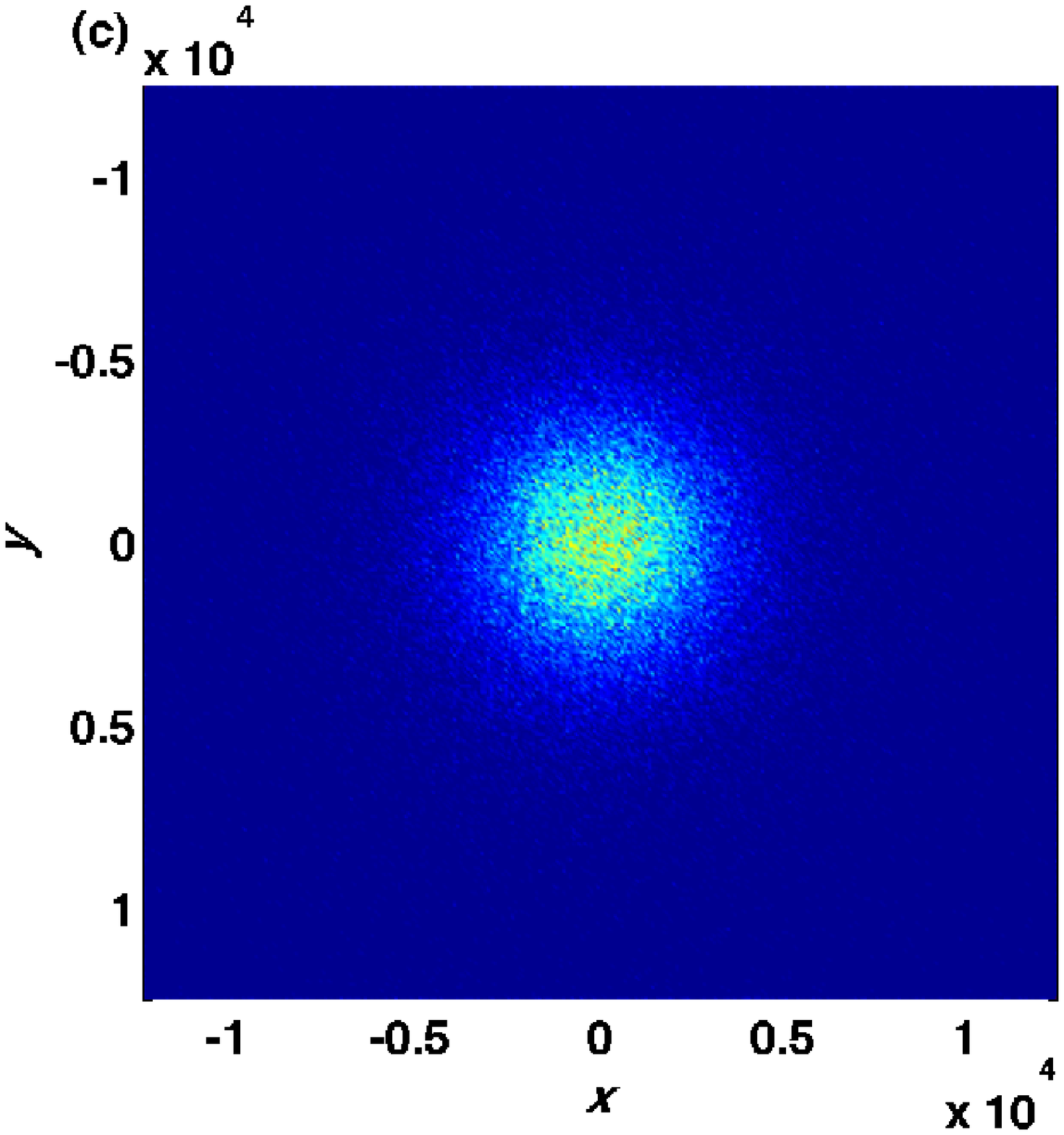}}
\epsfxsize=4cm
\subfigure{\epsfbox{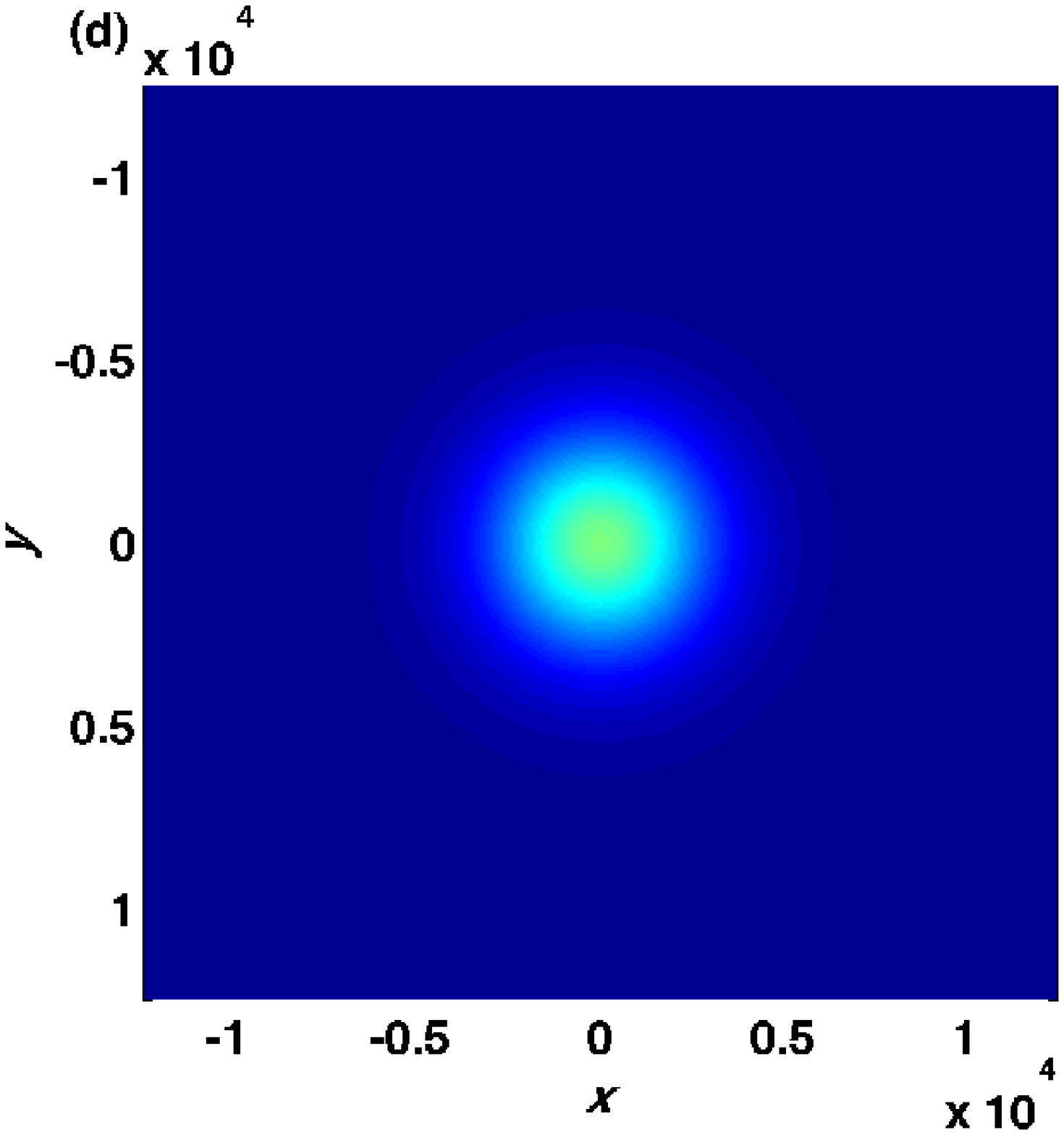}}
  \end{center}
  \caption{Comparison of simulations of Algorithm
    \ref{algo_vj_no_rests}, see Appendix \ref{app_gill}, and analytic
    solution to $m^0(t,\xv)$ given by equation
    \eqref{eq_final_finite_mean_sol}. Top row: case where $\kappa =0
    $, (a) simulation; (b) analytic solution. Bottom row: case where
    $\kappa =10$, (c) simulation; (d) analytic solution. Full details
    in main text.}
  \label{fig_sim_analysis_comp_1}
 \end{figure}

  \begin{figure}
  \begin{center}
\includegraphics[scale=.4]{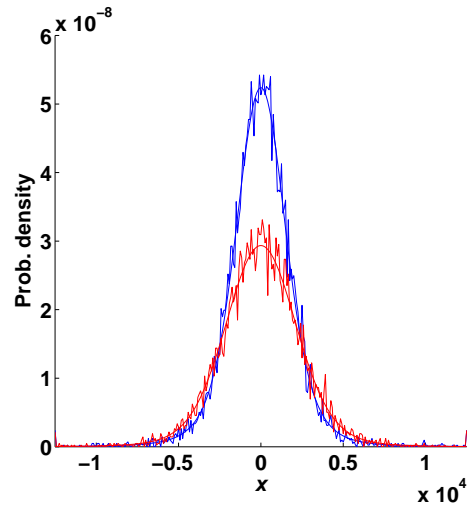}
 \end{center}
  \caption{Cross section of $m^0(t,\xv)$ for varying $\kappa$. For $\kappa = 0$, the smooth blue curve shows the analytic solution for $m^0$; the jagged blue curve shows the simulation constructed $m^0$. For $\kappa = 10$, the smooth red curve shows the analytic solution for $m^0$; the jagged red curve shows the simulation constructed $m^0$. Full details in main text.}
  \label{fig_sim_analysis_comp_2}
\end{figure}

\section{Discussion and conclusions}\label{sec_conc}

In this paper we have seen that for a generalised VJP, if we send the
second moment of the running distribution to infinity, in the limit of
large time we approximately obtain a fractional diffusion equation.
This is particularly interesting, since for a running distribution with
finite mean $\mu$, and variance $\sigma^2$, as detailed in
\cite{TaylorKing_2015}, we obtain an effective diffusion equation in
the large time limit with coefficient
\begin{equation}\label{eq_diff_const_fin_mean_fin_var}
D_{\text{eff}}= \frac{S_h^2 \mu}{n}  \left[ \frac{1}{1 - \psi_d} + \frac{1}{2}\left(\frac{\sigma^2}{\mu^2} - 1\right) \right] . 
\end{equation}
When the running time distribution is exponentially distributed, the
second term in the square brackets is identically zero. Therefore we
can view our diffusion constant as the contribution from the
exponential component of the running time distribution, plus an
additional (non-Markovian) term for non-exponential running times.
When considering equation \eqref{eq_final_finite_mean}, we see
essentially a Markov process, plus a non-Markov fractional correction
term. Additionally for equation \eqref{eq_final_finite_mean}, one
finds that in the limit as $\alpha\to1$, the diffusion equation is
recovered with diffusion constant given by equation
\eqref{eq_diff_const_fin_mean_fin_var}.

While our theory thus captures non-Markovian memory effects, it does
not carry through the finite velocity constraint of LWs as
implemented, e.g., by equation \eqref{eq:lwconstraint} to our
fractional diffusion equation \eqref{eq_final_finite_mean}. This
implies that its solution, the PDF equation
\eqref{eq_final_finite_mean_sol}, does not reproduce the ballistic
fronts and respective cut-offs of LW PDFs \cite{ZDK15}, see also
Figs.~2--3. Consequently it generates infinite second moments as for
L\'{e}vy flights. This is due to the expansion leading from equation
\eqref{eq_expand_fin_mean_inf_var} to equation
\eqref{eq_BinThm_expansion}, which is necessary in order to perform an
inverse Fourier-Laplace transform, as well as the Cattaneo
approximation. Fig. 3 suggests that for large times these fronts are
approximately negligible for reproducing the overall shape of LW PDFs.
Capturing them would necessitate generalising our theory to include
higher order moment equations, and avoiding the use of the Binomial
expansion in equation \eqref{eq_expand_fin_mean_inf_var}, which we
  performed to achieve analytical tractability.  However, the focus
of our present theory is not on these fronts but rather on the novel
type of intermittency emerging from LWs with memory, and the
possibility to go to higher dimensions.

One-dimensional versions of equation \eqref{eq_final_finite_mean}
  consisting of a sum of a conventional Laplacian modeling Brownian
  motion plus a fractional one reproducing L\'{e}vy flights have been
  written down \emph{ad hoc} in previous literature based on physical
  reasoning. This demonstrates the physical significance of our
  derivation and its result: Lomholt et al.\ \cite{LAM05} proposed an
equation of this type in order to model the optimal target search on a
fast-folding polymer chain by an ensemble of proteins. Here the
conventional Laplacian reproduced the one-dimensional diffusive
sliding of proteins or enzymes along the DNA chain while the
fractional Laplacian captured the intersegmental transfers, or jumps,
at chain contact points due to polymer looping. A generalised version
of this model was considered by Lomholt et al.\ \cite{LKMK08} in order
to study an intermittent search process that switches between local
Brownian search events and L\'{e}vy relocation times. On a purely
mathematical level, equations of this type form a subclass of
distributed-order fractional diffusion equations; see
Ref.~\cite{CSK11}, equations (39), (41) for modelling a diffusion
equation with a superposition of two (fractional) Laplacians, and its
solution equation (42) in terms of a convolution.  We remark that this
type of intermittent dynamics is different from the one considered by
B\'enichou et al.\ \cite{BLMV11}: There it was explicitly
distinguished between a Brownian search mode during which a target
could be found, and relocation dynamics during which a searcher was
insensitive for any target search. This dynamics was modelled by a set
of two coupled equations, with a different one for each process.  In
that sense, equations like equation \eqref{eq_final_finite_mean} are
somewhat closer to the concept of composite Brownian motion
\cite{Benh99,PMC13}: This dynamics was designed to model the search of
a forager, or particle, in patchy environments \cite{ZoLi99}, where
inter- and intra-patch movements were defined by Brownian motion with
different mean step lengths. This stochastic motion was generalised by
Reynolds \cite{Reyn09} in terms of an adaptive LW encompassing
composite Brownian motion, where the Brownian inter-patch movements
were replaced by a LW. Note, however, that in all the above models the
intermittent dynamics is put in by hand as a sum of two different
stochastic processes reproducing local search and non-local relocation
events while our equation \eqref{eq_final_finite_mean} emerges from a
single auto-correlated LW: Interestingly, here the Brownian term
is due to short-range auto-correlations in the LW dynamics while
the L\'{e}vy term results from the power law jumps. Our
  mathematical derivation thus gives all terms in equations which are
  of the type of equation \eqref{eq_final_finite_mean} precise
  physical meaning.

  This is important especially in view of a number of recent
  experiments: For the dinoflagellate \emph{Oxyrrhis marina} Bartumeus
  et al.\ \cite{BPPMC03} reported a switching between L\'{e}vy and
  Brownian search strategies depending on the density of its prey
  distribution. Similar results were obtained for coastal jellyfish
  \emph{(Rhizostoma octopus)} \cite{HBDF12}.  Movement patterns of
  crawling mussels \emph{(Mytilus edulis)} shifted from L\'{e}vy
  towards Brownian motion with increasing mussel density, where the
  Brownian motion emerged from frequent encounters between the mussels
  in dense environments \cite{dJBKW13}. Such type of intermittency can
  also be generated by a predator due to spatio--temporal sampling of
  prey in different environments: Sims et al.\ found a switching
  between L\'{e}vy and Brownian search patterns for a variety of
  free-ranging marine predators, where the animals were hunting either
  for sparse prey in deep ocean or for abundant prey close to the more
  productive shelf-edge \cite{SHBB12,Sims10}. In all these works the
  experimental data was analysed in view of either Brownian or
  L\'{e}vy dynamics but not by a superposition of both. Our new
  diffusion equation \eqref{eq_final_finite_mean} allows for the
  latter analysis by shedding light on the origin of this dynamics in
  terms of correlated LWs. Important for experimental applications is
  also that the (generalised) diffusion coefficients equation
\eqref{eq:diffcoeffs} quantifying this dynamics can be extracted from
measuring the PDFs for speed, running time, and turning angle.  That
our approach explicitly includes correlations is especially promising
for describing the movement of biological organisms, where memory
often matters \cite{LICCK12,LCK13}.  Along these lines it would be
interesting to derive a fractional Klein-Kramers equation that is more
general than equation \eqref{p_forward_no_rests} by containing
long-range correlations beyond two successive steps.  Following our
approach, one may then try to extract a fractional diffusion equation
for a long-range correlated L\'{e}vy walk. In terms of more general
applications, we note that an intermittent switching between L\'{e}vy
and Brownian search may be advantageous to optimise the random search
of a mobile robot for adapting efficiently under changing target
density \cite{NMNS10}.

\begin{acknowledgments}
  We would like to thank Thomas Hillen and Jon Chapman for helpful
  discussions.  J.P.T-K.\ received funding from the EPSRC under grant
  reference number EP/G037280/1. S.F.\ acknowledges the support of the
  EPSRC Grant No. EP/J019526/1, R.K.\ thanks the Office of Naval
  Research Global for finanical support.
\end{acknowledgments}


\appendix

\section{Gillespie algorithm}\label{app_gill}

Here we give a simple Gillespie algorithm \cite{Gillespie_1977} for
generating a sample path up until time $T_{\text{end}} > 0$. It
should be noted that the sample path will need to be truncated as the
algorithm generates positions beyond $T_{\text{end}}$.

\begin{algorithm}[H]
\DontPrintSemicolon
 \KwData{Initialise time $t=0$, starting position at $\vect{x}(t=0) = \vect{x}_0$ and starting velocity at $\vect{v}(t=0) =\vect{v}_0$.}
 Assume particle has just initiated a running state. \;
 \While{$t < T_{\text{end}}$}{
  Sample time spent running $\tau\sim f_{\tau}(t)$.\;
  Update position: $\vect{x}(t+\tau) \leftarrow \vect{x}(t) + \tau\vect{v}(t)$.\;
  Sample new velocity for next running phase: $\vect{v}(t+\tau)\sim T(\cdot, \vect{v}(t))$.\;
  Update time $t \leftarrow t + \tau $.\;
  }
 \caption{Algorithm to generate a single generalised VJ sample path without rests.}
 \label{algo_vj_no_rests}
\end{algorithm}

\section{Alternative derivation}\label{app_deriv}

We present now an alternative way to derive the main governing
equation \eqref{p_forward_no_rests} for the density $p(t,\xv,\vv)$,
using Alt's structural approach \cite{Alt_1980}. One can consider the
motion of an individual (bacteria, cell, etc.) that runs with the
velocity $\vv$ during the run time $\tau$ and stops at $(t,\xv)$ with
given probability $\beta _{r}(\tau)$ per unit time. We define the mean
structural density of individuals, $\sigma (t,\xv,\vv,\tau )$, at
point $\xv$ and time $t$ that move with the velocity $\vv$ and having
started the move $\tau$ units of time ago. The governing equation for
$\sigma (t,\xv,\vv,\tau)$ takes the form \cite{Alt_1980}
\begin{equation}
\frac{\partial \sigma }{\partial t}+\vv \cdot \nabla_{\xv} \sigma +\frac{\partial \sigma }{\partial \tau }=-\beta _{r}(\tau )\sigma  \, .  \label{eq_str}
\end{equation}
We assume that at the initial time $t=0$ all individuals have zero running
time
\begin{equation}
\sigma (0,\xv,\vv,\tau )=p_{0}(\xv,\vv)\delta (\tau ) \, ,  \label{eq_initial}
\end{equation}
where $p_{0}(\xv,\vv)$ is the initial density. Our purpose is
to obtain the master equation for the density
\begin{equation}
p(t,\xv,\vv)=\int_{0}^{t}\sigma (t,\xv,\vv,\tau)\dd\tau \, .  \label{eq_den}
\end{equation}
We set up the boundary condition at zero running time $\tau =0$:
\begin{equation}
\sigma (t,\xv,\vv,0)=\int_{0}^{t}\beta _{r}(\tau)\int_{V}T\left( \vv,\vv' \right) \sigma (t, \xv,\vv' ,\tau )\dd\vv' \dd\tau \, .  \label{eq_in0}
\end{equation}
The master equation for $p(t,\xv,\vv)$ can be found by
differentiating equation \eqref{eq_den} with respect to time $t$
\begin{equation}
\frac{\partial p}{\partial t}+\vv \cdot \nabla_{\xv} p=-i(t,\xv, \vv)+j(t,\xv,\vv) \, ,  \label{eq_eq1}
\end{equation}
where the switching terms are
\begin{equation}
i(t,\xv,\vv)=\int_{0}^{t}\beta _{r}(\tau )\sigma (t,\xv,\vv,\tau )\dd\tau \, ,\quad j(t,\xv,\vv)=\sigma (t,\xv,\vv,0) \, .  \label{eq_sw}
\end{equation}
By using the method of characteristics we find from equation \eqref{eq_str} for $\tau < t$
\begin{equation}
\sigma (t,\xv,\vv,\tau )=\sigma (t-\tau ,(\xv-\vv)\tau,\vv,0) \exp \left\{ -\int_{0}^{\tau }\beta _{r}(s)\dd s \right\}  .  \label{eq_solution1}
\end{equation}
The exponential factor in the above formula is the survival function
\begin{equation}
F_{\tau }(t)=\exp \left\{ {-\int_{0}^{t}\beta _{r}(s)\dd s} \right\} \, .  \label{eq_Sur}
\end{equation}
By using equation \eqref{eq_solution1} and the Laplace transform technique \cite{Fedotov_2013, Fedotov_2015,Straka_2015}, we find the expressions for the switching terms:
\begin{equation}
i(t,\xv,\vv)=\int_{0}^{t}\Phi _{\tau }(t-s)p(s,\xv-\vv(t-s),\vv)\dd s \, ,  \label{eq_i}
\end{equation}
\begin{equation}
j(t,\xv,\vv)=\int_{0}^{t}\Phi _{\tau }(t-s)\int_{V}T\left(
\vv,\vv' \right) p(s, \xv - (t-s)\vv',\vv' )\dd\vv' \dd s \, .  \label{eq_j}
\end{equation}
The main advantage of the present derivation is that it can be easily
extended for the nonlinear case \cite{Fedotov_2013}. Superdiffusive equations can be
obtained for the following rate \cite{Fedotov_2013_b, Ferrari_2001}
\begin{equation}
\beta _{r}(t)=\frac{\alpha }{\tau _{0}+t},\ 0<\alpha <2\, .  \label{eq_inverse}
\end{equation}
The rate equation \eqref{eq_inverse} leads to a power law (Pareto) survival function
\begin{equation}
F_{\tau }(t)=\left[ \frac{\tau _{0}}{\tau _{0}+t}\right] ^{\alpha } \, ,
\label{eq_Pareto}
\end{equation}
and corresponding running time PDF
\begin{equation}
f_{\tau }(t)=\frac{\alpha \tau _{0}^{\alpha }}{(\tau _{0}+t)^{1+\alpha }} \, .
\label{eq:psi_tails}
\end{equation}

\section{Laplace space expansion of Pareto power law distribution}\label{app_Pareto}

Consider the Pareto power law distribution with parameters $\tau_0$ and $\beta$. 
\begin{equation}\label{eq_f_tau_pareto}
f_\tau ( t )= \frac{\beta \tau _{0}^{\beta }}{(\tau _{0}+ t )^{1+\beta }} \, \iff \, \bar{f}_\tau (\lambda) =  \beta \left( \tau _{0}\lambda\right) ^{\beta }e^{\tau _{0} \lambda}\Gamma
\left( -\beta ,\tau _{0}\lambda \right)\, 
\end{equation}
where we used the incomplete gamma function $\Gamma (x,y) := \int_{y}^{\infty }t^{x-1}e^{-t} \dd t$. The mean and variance are both infinite for $0<\beta<1$, but the distribution has finite mean for $1<\beta<2$. Using the asymptotic expansion
\begin{equation}
\Gamma \left( - \beta ,y  \right) =-\frac{\Gamma (1-\beta )}{\beta }+ 
y^{-\beta }\beta^{-1}+\frac{ y ^{1 - \beta } }{1 - \beta }+... \quad \text{as }y\rightarrow 0\, ,
\end{equation}
we recover an expansion of the form given in equation \eqref{eq_f_tau_Laplace_exp_fin_mean} where $\mu = \tau_0/\alpha$ and $\gamma = -\tau_0^{1 + \alpha}\Gamma(-\alpha)$ for $\beta = 1 + \alpha$.

\section{Approximation of fractional term}\label{app_riesz}

By closing the set of moment equations as explained in
  Sec.~\ref{sec_cattaneo}, in the fractional case one obtains terms of
  the form $\int_V \vv (\vv \cdot \nabla_{\xv})^\beta p\, \dd \vv$. To
  evaluate these terms requires the use of the Riesz derivative, which
  we define first.

\begin{Definition}

In the Meerschaert's framework \cite{DOvidio_2014, Meerschaert_2004, Meerschaert_2006}, for scalar function $f:\mathds{R}^n \rightarrow \mathds{R}$ the multidimensional fractional derivative is given by
\begin{equation}\begin{aligned}
\nabla_M^\beta f(\xv) & = \int_{\vectornorm{\vect{\theta}} = 1} \vect{\theta} D_{\vect{\theta}}^\beta f(\xv) M(\dd \vect{\theta})\\ & = \int_{\vectornorm{\vect{\theta}} = 1} \vect{\theta} (\vect{\theta}  \cdot \nabla_{\xv})^\beta f(\xv) M(\dd \vect{\theta})\, ,  
\end{aligned}\end{equation}
for $\xv\in\mathds{R}^n$ and $\beta\in(0,1)$, where $\vect{\theta}= (\theta_1,...,\theta_n)$ is a unit column vector. We require that $M(\dd \vect{\theta})$ is positive finite measure, called a mixing measure. The term $(\vect{\theta}  \cdot \nabla_{\xv})^\beta$ is called the $\beta$ order fractional directional derivative given by
\begin{equation}
\mathscr{F} \left\{  (\vect{\theta}  \cdot \nabla_{\xv})^\beta f(\xv) \right\} =  (i\vect{\theta}  \cdot \vect{k})^\beta \tilde{f}(\vect{k}) \,  . 
\end{equation}
For the vector valued function $\vect{f}:\mathds{R}^n \rightarrow \mathds{R}$, we define the fractional gradient by
\begin{equation}\begin{aligned}
J_M^\beta \vect{f}(\xv) & = \int_{\vectornorm{\vect{\theta}} = 1} \vect{\theta} D_{\vect{\theta}}^\beta (\vect{\theta} \cdot \vect{f}(\xv) ) M(\dd \vect{\theta}) \\
& = \int_{\vectornorm{\vect{\theta}} = 1} \vect{\theta} (\vect{\theta}  \cdot \nabla_{\xv})^\beta  (\vect{\theta} \cdot \vect{f}(\xv) ) M(\dd \vect{\theta})\, ,
\end{aligned}\end{equation}
where $\xv\in\mathds{R}^n$ and $\beta\in(0,1)$. In the case when $M(\dd \vect{\theta}) = \text{const }\dd \vect{\theta}$, we get the Riesz derivative. For the remainder of this paper, we always assume this constant is identically one. We then define the fractional Laplacian for scalar function $f$ by
\begin{equation}\label{eq_Riesz_laplacian}
\mathds{D}^{1+ \beta}_M f(\xv) = \nabla_{\xv}  \cdot \nabla_M^{\beta}  f( \xv ) 
= \int_{\vectornorm{\vect{\theta}} = 1}  (\vect{\theta}  \cdot \nabla_{\xv})^{1 + \beta} f(\xv) M(\dd \vect{\theta}) \, .
\end{equation}
Additionally, for the vector valued function $\vect{f}$, we write
\begin{equation}
\mathds{D}^{1+ \beta}_M \vect{f}(\xv) = \nabla_{\xv}  \cdot J_M^{\beta}  \vect{f}( \xv ) 
 = \int_{\vectornorm{\vect{\theta}} = 1}  (\vect{\theta}  \cdot \nabla_{\xv})^{1 + \beta} (\vect{\theta} \cdot \vect{f}(\xv) ) M(\dd \vect{\theta}) \, .
\end{equation}

\end{Definition}

We now wish to evaluate terms of the form
\begin{equation}
\int_V \vv(\vv \cdot\nabla_{\xv })^{\alpha}p(t, \xv, \vv)\dd \vv \, .
\end{equation}
From equation \eqref{eq_u_min_form}, we gave a form for $u_{\text{min}}(t,\xv, \vv)$. This allows us to evaluate the fractional flux term as
\begin{align}\label{eq_throw_term_eval}
& \quad \int_V  \vv(\vv \cdot\nabla_{\xv })^{\alpha}u_{\text{min}}\dd \vv \\
=& \int_V \vv(\vv \cdot\nabla_{\xv })^{\alpha}\left[ \frac{m^0(t, \xv ) h(s)}{s^{n-1}A_{n-1}} + \frac{(\vect{m}^1(t, \xv ) \cdot \vv ) h(s)}{S_h^2s^{n-1}V_n} \right] \dd \vv \, ,  \nonumber 
\end{align}
Evaluating equation \eqref{eq_throw_term_eval} in polar/spherical coordinates, one obtains
\begin{align}
& \quad \int_V  \vv(\vv \cdot\nabla_{\xv })^{\alpha}u_{\text{min}}\dd \vv \\
&=\frac{1}{A_{n-1}} \int_0^\infty s^{1 + \alpha} h(s) \dd s \int_{\vectornorm{\vect{\theta}}=1}  \vect{\theta}(\vect{\theta} \cdot\nabla_{\xv })^{\alpha} m^0(t, \xv )  \dd\vect{\theta} \, ,\nonumber \\
&+ \frac{1}{S_h^2 V_n} \int_0^\infty s^{2+\alpha} h(s)  \dd s  \int_{\vectornorm{\vect{\theta}}=1}  \vect{\theta}(\vect{\theta} \cdot\nabla_{\xv })^{\alpha} (\vect{m}^1(t, \xv ) \cdot \vect{\theta} )  \dd\vect{\theta}   \, , \nonumber  \\
&\quad\quad = \frac{ S_h^{1 + \alpha }}{ A_{n-1}  }  \nabla_{M}^{\alpha} m^0 + \frac{S_h^{2 + \alpha} }{V_n S_h^2 } J^{\alpha}_M \vect{m}^1 \, .  \nonumber 
\end{align}

\section{Proof of Theorem \ref{Thm_relate_frac_deriv}}\label{app_proof}

We wish to prove
\begin{equation}
\mathds{D}^{1 + \alpha}_M f(\xv) \equiv \Upsilon_d(\alpha) \Delta^{\frac{1 + \alpha}{2}}_{\xv} f(\xv) \, ,
\end{equation}
for $M(\dd \vect{\theta}) = \dd \vect{\theta}$. We first consider the left hand side. In Fourier space, this is
\begin{equation}\label{eq_pr_ln_1}
\mathds{D}^{1+ \beta}_M \tilde{f}(\kv)  = \int_{\vectornorm{\vect{\theta}} = 1}  (i \vect{\theta}  \cdot \kv )^{1 + \beta} \tilde{f}(\kv) \dd \vect{\theta} \, .
\end{equation}

In the two-dimensional case, rewriting equation \eqref{eq_pr_ln_1}, using the substitution $\kv = || \kv || (\cos \psi, \sin\psi)$ [for $\psi\in[0,2\pi)$] and using polar coordinates, we identify $\Upsilon_2(\alpha)$ as
\begin{equation}\label{eq_pr_ln_2}
\Upsilon_2(\alpha) = -\int_0^{2\pi} \left\{ 
i
\left(\begin{array}{c}
\cos \psi \\ \sin \psi
\end{array}\right) \cdot
\left(\begin{array}{c}
\cos \theta \\ \sin \theta
\end{array}\right)
 \right\}^{1 + \alpha} \dd\theta \, .
\end{equation}
The negative sign appears from the definition of the fractional Laplacian [given in equation \eqref{eq_def_frac_Laplace}]. Using trigonometric identities, we manipulate equation \eqref{eq_pr_ln_2}, finding
\begin{equation}\label{eq_pr_ln_3}\begin{aligned}
\Upsilon_2(\alpha) & = - \int_0^{2\pi} \left\{ i \cos(\theta - \psi) \right\}^{1 + \alpha} \dd\theta \\
& = -2\left[ (i)^{1 + \alpha} + (-i)^{1 + \alpha} \right] \int_0^{\pi/2} \left\{ \cos(\eta) \right\}^{1 + \alpha} \dd \eta \\
& = 4\sin\left( \frac{\alpha \pi}{2}\right)   \left\{ \frac{\sqrt{\pi}}{2}  \frac{ \Gamma (\frac{2 + \alpha}{2})}{\Gamma (\frac{3 + \alpha}{2})}  \right\} \, ,
\end{aligned}\end{equation}
and the two-dimensional case is proved.

For the three-dimensional case, we wish to evaluate an integral similar to equation \eqref{eq_pr_ln_2}. Using the representation $\kv = || \kv || (\sin\psi_1  \cos\psi_2 , \sin\psi_1  \sin\psi_2 , \cos\psi_1)$ [for $\psi_1\in(0,\pi)$, $\psi_2\in(0,2\pi)$], and using spherical coordinates this integral can be evaluated similar to that of $\Upsilon_2(\alpha)$.

\end{document}